\documentclass[11pt,a4paper,english]{amsart}
\usepackage{babel}
\usepackage[latin1]{inputenc}
\usepackage{amsmath}
\usepackage{amsthm}
\usepackage{amsfonts}
\usepackage{indentfirst}
\usepackage{graphicx}
\usepackage{amssymb}
\usepackage[mathscr]{eucal}
\usepackage{breqn}
\usepackage{tikz}
\usepackage{caption}

\usepackage{enumerate}
\usepackage{amsthm}
\usepackage{amscd}
\newtheorem{teo}{Theorem}

\newtheorem{pro}[teo]{Proposition}
\newtheorem{lem}[teo]{Lemma}

\theoremstyle{definition}
\newtheorem{rem}[teo]{Remark}
\newtheorem{de}[teo]{Definition}
\newtheorem{exa}{Example}

\newcommand{\fq}{\mathbb{F}_q}

\title[EAQECC From RS Codes and BCH Codes with Extension Degree 2]{Entanglement-Assisted Quantum Error Correcting Codes From RS Codes and BCH Codes with Extension Degree 2}
\author{Carlos Galindo, Fernando Hernando and Diego Ruano}
\curraddr{\texttt{Carlos Galindo and Fernando Hernando:} Instituto
Universitario de Matem\'aticas y Aplicaciones de Castell\'on and
Departamento de Matem\'aticas, Universitat Jaume I, Campus de Riu
Sec. 12071 Castell\'{o} (Spain)\\
\texttt{Diego Ruano:} IMUVA-Mathematics Research Institute, Universidad de Valladolid, 47011 Valladolid (Spain).
}
\email{{\rm Galindo:} galindo@uji.es;  {\rm Hernando:} carrillf@uji.es; {\rm Ruano:} diego.ruano@uva.es}
\date{}
\thanks{This work was supported in part by the Spanish MICINN/FEDER grants PGC2018-096446-B-C21, PGC2018-096446-B-C22 and RED2018-102583-T, by the Spanish MINECO grant RYC-2016-20208 (AEI/FSE/UE), by the Junta de CyL (Spain) grant VA166G18, by the Generalitat Valenciana (Spain) grant AICO-2019-223 and by Universitat Jaume I grant P1-1B2018-10.}


\begin{document}

\begin{abstract}
Entanglement-assisted quantum error correcting codes (EAQECCs) constructed from Reed-Solomon codes and BCH codes are considered in this work. It is provided a complete and explicit formula for the parameters of EAQECCs coming from any Reed-Solomon code, for the Hermitian metric,  and from any BCH code with extension degree $2$ and consecutive cyclotomic cosets, for both the Euclidean and the Hermitian metric. The main task in this work is the computation of a completely general formula for $c$, the minimum number of required maximally entangled quantum states.
\end{abstract}

\maketitle

\section{Introduction}\label{se:uno}

Quantum error correcting codes (QECCs) are mostly defined using classical linear codes \cite{Calderbank,Gottesman}. They were first introduced over the binary field and then the construction was extended to an arbitrary finite field \cite{Ketkar}. Namely, QECCs over $\mathbb{F}_q$, the finite field with $q$ elements, are usually constructed from self-orthogonal classical codes. We can consider classical codes over $\mathbb{F}_q$ if we decide to use the Euclidean metric, and over $\mathbb{F}_{q^2}$, when using the Hermitian metric. 

Brun, Dvetak and Hsieh in \cite{Brun} proposed to share entanglement between encoder and decoder to simplify the theory of quantum error-correction and increase the communication capacity, giving rise to entanglement-assisted quantum error correcting codes (EAQECCs). An important advantage of this new construction is that one may consider an arbitrary classical linear code, without the self-orthogonality restriction. EAQECCs were also first defined over the binary field and then the construction was extended to an arbitrary finite field \cite{QINP3}.

The main difficulty for determining the parameters of an EAQECC, with respect to QECCs,  is the computation of the parameter $c$, the minimum number of required maximally entangled quantum states in $\mathbb{C}^q \otimes \mathbb{C}^q$. Computationally speaking, given a concrete EAQECC, the computation of $c$ is not intense but it is a difficult task to provide a formula for a given family of codes. This is the main goal of the article.

Several articles have explored the construction of EAQECCs from classical linear codes as binary BCH codes \cite{Lu,Lv,Mesnager}, constacyclic codes \cite{Liu,Sari}, constacylic LCD codes \cite{Qian}, cyclic codes \cite{Quian2},  generalized Reed-Solomon codes \cite{Guo2,Guo3,Guo}, negacyclic BCH codes \cite{Chen,Chen2} and algebraic-geometry codes \cite{Pereira2}. These papers address particular cases and determine their parameters but no general formula for the parameters of EAQECCs using the previous families of codes is known. A more general procedure is given in  \cite{Pereira}, which provides parameters for EAQECCs obtained from Reed-Solomon (RS) codes both for the Euclidean and Hermitian metric, however not all RS codes in the Hermitian case are considered. In this regard and in this paper, we give a formula for the parameters of all the EAQECCs that can be obtained from RS codes with respect to the Hermitian metric (see Section \ref{se:cuatro}).

The main aim of this article is to consider BCH codes over an arbitrary field $\mathbb{F}_q$, that can be understood as cyclic codes but also as subfield subcodes of RS codes. We consider BCH codes defined from consecutive cyclotomic cosets and hence we can bound their minimum distance by the well-known BCH bound. We give a complete and explicit formula for the parameters of the EAQECCs coming from {\it any} BCH code (defined by consecutive cyclotomic cosets) with extension degree 2, both for the Euclidean and the Hermitian metric. Moreover, we also determine the parameters when we extend the classical codes by evaluating at zero. With our formulae the reader can easily determine the parameters of the EAQECCs obtained as mentioned.

The computation of $c$ is performed following the geometric decomposition of a linear code \cite{Diego}, that was stated in \cite{QINP3} for EAQECCs. The computations are carried out by a careful analysis of the involved cyclotomic cosets and the $q$-adic decomposition of its elements that allow us to compute the geometric decomposition of a linear code. The geometric decomposition of a linear code endowed with a inner product provides a basis of the code included in a specific basis of the ambient linear space. This last basis $\{v_i\}_{i\in I}$ satisfies that each $v_i$ is either orthogonal to $v_j$ for all $j \neq i$ (symmetric) or  orthogonal to $v_j$ for all $j$ except for a unique reciprocal element $v_k$ with $k \neq i$ (asymmetric). This means that one can determine the hull, or radical, of the linear code easily. In this article, we compute such a decomposition by identifying symmetric and asymmetric cosets (see Definition \ref{def:syasi}) with respect the inner product considered. We notice that the method is similar to the one in \cite{Lu} for quaternary BCH codes.

In Section \ref{se:cinco}, we give a general formula for the parameters of EAQECCs coming from BCH codes with the Euclidean metric and extension degree 2.  We show that they have parameters that are worse than the ones obtained with RS codes with Hermitian metric (studied in Section \ref{se:cuatro}). Namely, they have the same length, $q^2-1$ or $q^2$, and bound for the minimum distance but the dimension is lower for BCH codes with the Euclidean metric. However, BCH codes with the Euclidean metric have still some interest since one may obtain codes whose parameter $c$ is greater than the one obtained with RS codes. In this way, we increase the constellation of known EAQECCs that is  limited at this moment. Moreover, the computation of $c$ for BCH codes with the Euclidean metric helps understanding the computation of $c$ for BCH codes with the Hermitian metric treated in Section \ref{se:seis}, that is rather technical.

Finally, the main results of this article are given in Section \ref{se:seis}, where BCH codes with the Hermitian metric and extension degree 2 are considered. We give a {\it completely} general formula for their parameters. Furthermore, using that formula we present some EAQECCs with good parameters, by giving tables of them over different finite fields whose parameters exceed the Gilbert-Varshamov (GV) bound \cite{QINP3}. We also compare their parameters with the codes available in the literature, when it is possible. We obtain long codes over $\mathbb{F}_q$, with length $q^4-1$ and $q^4$, and with good parameters $[[n,k,d;c]]_q$. Our codes satisfy $k >c$ and thus, they give rise to catalytic quantum codes \cite{catalytic}.  Note that one cannot consider the Singleton bound, $n+2 \ge k + 2d + c$,  for EAQECCs defined over a non-binary field since that bound is only proved for codes defined over a binary field with $d \le (n+2)/2$ \cite{Lai}, despite several articles in the literature consider it for arbitrary fields. 

The article is organized as follows: we introduce RS and BCH codes in Section \ref{se:dos} and EAQECCs coming from them in Section \ref{se:tres}. Section \ref{se:cuatro} is devoted to the study of EAQECCs from RS codes with respect to the Hermitian metric. To conclude, EAQECCs obtained from BCH codes with extension degree 2 are presented in Section \ref{se:cinco} in the Euclidean case and in Section \ref{se:seis} for the Hermitian case.

\section{RS and BCH codes}\label{se:dos}

We introduce RS codes and BCH codes in this section. We regard BCH codes as subfield subcodes of evaluation codes as in \cite{Bier,Cas} instead of as cyclic codes. This construction has the advantage that it can be extended to evaluation by polynomials in several variables.

Let $p$ be a prime number and consider the finite field $\mathbb{F}_{p^\ell}$ with $p^\ell$  elements. Let $n = p^\ell-1$ and $\mathbb{F}_{p^\ell}[X]$ the ring of polynomials in one variable with coefficients in $\mathbb{F}_{p^\ell}$. Consider classes of univariate polynomials in the quotient ring $\mathbb{F}_{p^\ell} [X]/J$, where $J$ is the ideal of $\mathbb{F}_{p^\ell}[X]$ generated by  $X^{n} -1$. Define $$\mathrm{ev}: \mathbb{F}_{p^\ell}[X]/J \rightarrow \mathbb{F}_{p^\ell}^{n}; ~  \;\; \mathrm{ev}(f) = \left(f(P_1), f(P_2), \ldots, f(P_{n}) \right),$$where $Z=\{P_1, P_2, \ldots, P_n\}$ is the zero locus of $J$ in $\mathbb{F}_{p^\ell}$. Let $\Delta$ be a subset of $\mathcal{H}:=\{0,1, \ldots, n-1 \}$. Then the RS code, $D_\Delta \subseteq \mathbb{F}_{p^\ell}^n$, is the code generated by $$\left\{ \mathrm{ev}\left(X^i\right) \; | \; i \in \Delta \right\}.$$Usually one considers $\Delta = \{0,1, \ldots , k-1\}$ and the RS code has parameters $[p^\ell-1,k,p^\ell-k]_q$.  Moreover, one can extend the previous code by evaluating at $0$ as well and therefore, one obtains a code with parameters $[p^\ell,k,p^\ell+1-k]_q$.

Let $r$ be a positive integer such that $r$ divides $\ell$. We consider first codes over the field $\mathbb{F}_{p^\ell}$ and then their subfield subcodes over the field $\mathbb{F}_{p^r}$. BCH codes can be defined as subfield subcodes of the form $D_\Delta \cap (\mathbb{F}_{p^r})^n$ and extend RS codes in the sense that one can consider that a RS code is a BCH with extension degree one, that is, $r=\ell$. This is why we may consider the same notation for both families of codes in this article.

In the ring $\mathbb{Z}_{n}$, we consider minimal cyclotomic cosets with respect to $q=p^r$, minimal means that it contains exactly the elements of the form $xq^t$, $t \geq 0$, in $\mathbb{Z}_{n}$ for some  $x \in \mathbb{Z}_{n}$ under the identification $\mathbb{Z}_{n} = \mathcal{H}$. We denote by $I_x$ the minimal cyclotomic coset $\{ xq^t: t \ge 0 \}$. For every minimal cyclotomic coset, pick its least element, then let $\mathcal{A}$ be the set of all minimal representatives and $\{I_{x}\}_{ x \in \mathcal{A}}$ is the set of minimal cyclotomic cosets with respect to $q$. Moreover, let $i_{x}  := \# (I_{x})$, with $\#$ denoting the cardinality of a set. For convenience, we write $$\mathcal{A} = \{m_0 = 0  < m_1 < m_2 < \cdots \} = \{m_j\}_{j=0}^z.$$

We will use the following two results which can be found in \cite{galindo-hernando,QINP2}.

\begin{pro}
\label{la7}
Set $\Delta = \cup_{j=t'}^t I_{m_j}$, $t' < t$. Then the subfield-subcode of $D_\Delta$ over $\mathbb{F}_q$,
$$
E_\Delta = D_\Delta |_{\mathbb{F}_{q}} = E_\Delta \cap (\mathbb{F}_{q})^n,
$$
has dimension $\sum_{j=t'}^{t}  i_{m_j}$.
\end{pro}

The forthcoming Proposition \ref{la8} uses duality with respect to the  two metrics we consider in this work. These metrics are induced by two inner products: the Euclidean inner product, where $x\cdot y = \sum_{i=1}^n x_i y_i$ and the source code is defined over the finite field with $q$ elements, and the Hermitian inner product, where $x\cdot y = \sum_{i=1}^n x_i y_i^q$ and the source code is defined over the finite field with $q^2=p^r$ elements (for a suitable $r$ that divides $\ell$); in this last case, we have to replace $q$ with $q^2$ in the above description. Both cases, will allow us to obtain EAQECCs over the finite field with $q$ elements by \cite{QINP3}.

For the sake of clarity, we consider a toy example through the article only for illustrating the main concepts and notation.

\begin{exa}\label{ej:1}
Let $p=2$, $\ell =4$ and $r=2$. We consider the Hermitian inner product, thus $q^2=p^r=4$ and $n=p^\ell -1=15$. The cyclotomomic cosets in $\mathbb{Z}_{15}$ with respect to $4$ are:

$$I_{0}=\{0\},~I_1=\{1,4\},~I_2=\{2,8\},~I_3 = \{3,12\},~I_5= \{5\},$$$$I_6=\{6,9\},~I_7 = \{7,13\},~I_{10}=\{10\},~I_{11}=\{11,14\}.$$

And therefore, the set of all minimal representatives is
$\mathcal{A}=\{ 0,1,2,3,5,6,7,10,11\}$.\qed
\end{exa}

\begin{pro}
\label{la8}
The minimum distance of the (Euclidean or Hermitian) dual code of $E_\Delta$, $C_\Delta$, where $\Delta = \cup_{j=0}^t  I_{m_j}$,  is larger than or equal to $m_{t+1}+1$ (BCH bound).
\end{pro}

Therefore for $\Delta = \Delta(t) = \cup_{j=0}^t  I_{m_j}$, the code $C_\Delta$ is known as a BCH code and it has parameters $[n, n - \sum_{j=0}^{t}  i_{m_j}, \ge m_{t+1}+1 ]_{p^r}$.

For the Euclidean case, set $I_x^\bot := I_{n-x}$ and define
\begin{eqnarray*}
\Delta(t)^\bot & := & (I_{m_0} \cup I_{m_1} \cup \cdots \cup I_{m_t})^\bot\\
 & := & \mathcal{H} \setminus (I_{m_0}^\bot \cup I_{m_1}^\bot \cup \cdots \cup  I_{m_t}^\bot)\\
 & = & \mathcal{H} \setminus (I_{n-m_0} \cup I_{n-m_1} \cup \cdots \cup  I_{n-m_t}).
\end{eqnarray*}

Analogously, for the Hermitian case with base field $\mathbb{F}_{q^2}$, we set $I_x^\bot := I_{n-qx}$ and $\Delta(t)^\bot := \mathcal{H} \setminus (I_{n-qm_0} \cup I_{n-qm_1} \cup \cdots \cup I_{n-qm_t})$. With the above notations, by \cite{QINP,QINP2}, it holds $$C_{\Delta (t)} = E_{\Delta (t)}^\bot = E_{\Delta(t)^\bot}.$$

\begin{exa}\label{ej:2}
This is a continuation of Example \ref{ej:1} with $p=2$, $\ell =4$, $r=2$, $q^2=4$ and $n=15$, we consider the Hermitian inner product. Let $t=6$, then $m_6=7$ and $$\Delta(6)=I_0\cup I_1\cup I_2\cup I_3\cup I_5\cup I_6\cup I_7.$$

Moreover, we have that \begin{eqnarray*}\Delta(6)^\bot & = & \mathcal{H} \setminus (I_{n-0q} \cup I_{n-1q} \cup  I_{n-2q} \cup I_{n-3q} \cup I_{n-5q}  \cup  I_{n-6q} \cup  I_{n-7q})\\
& = & \{0,1, \ldots , 14\} \setminus (I_{0} \cup I_{7} \cup  I_{11} \cup I_{6} \cup I_{5}  \cup  I_{3} \cup  I_{1})\\
& = & I_2 \cup I_{10}.
 \end{eqnarray*}
Therefore, $C_{\Delta (6)}=E_{\Delta(6)^\bot}$ has parameters $[15, 3, 11 ]_4$.\qed

\end{exa}

As for RS codes, we can extend BCH codes by evaluating at 0 as well. Then we obtain a code with parameters $[n+1, n+1 - \sum_{j=0}^{t}  i_{m_j}, \ge m_{t+1}+1 ]_{p^r}$. In this paper, we only consider extension degree 1 or 2 (i.e., $\ell=r$ or $\ell = 2r$) and the length $n$ of our codes will be $p^\ell-1$ or $p^\ell$.

\section{EAQECC}\label{se:tres}

We compute in this section parameters of entanglement-assisted quantum error correcting codes coming from RS codes and BCH codes. As we have mentioned in the previous section, for easing the notation, we consider RS codes as BCH codes with extension degree 1. That is, for RS codes one has that $r=\ell$, $n=p^\ell-1=p^r-1$, and all the cyclotomic cosets have size 1.

From Corollary 1 and Theorem 4 in \cite{QINP3}, we have 

\begin{teo}\label{th:eaqecc}
Let $E$ be a linear code over $\fq$ (over $\mathbb{F}_{q^2}$) with length $n$ and dimension $k$ and let $C$ be its Euclidean (Hermitian) dual code that has minimum distance $d$. Then there exists an EAQECC with parameters $$[[n,n-2k+c,d;c]]_q,$$where $c= \dim E - \dim (E \cap C)$.
\end{teo}

This implies that, for $\Delta = \Delta(t)= \cup_{j=0}^t  I_{m_j}$,  there exists an EAQECC with parameters $$[[n,n-2\sum_{j=0}^{t}  i_{m_j}+c,\ge m_{t+1}+1;c]]_q.$$ Hence, the only task for completely determining its parameters remains computing $c$.

It is not feasible to give a formula for $c$ in the general case. The aim of this article is to provide closed formulas for extension degree equal to 1 (RS codes) and 2. Thus, assuming that we do not evaluate at 0, in the first case $\ell=r$, $n=q-1$ when considering codes over $\mathbb{F}_q$ and Euclidean duality and $n=q^2-1$ when we use codes over $\mathbb{F}_{q^2}$ and Hermitian duality, and in the second case $\ell = 2r$, $n=q^2-1$ when we consider subfield subcodes over $\mathbb{F}_q$ from codes over $\mathbb{F}_{q^2}$ and Euclidean duality and $n=q^4-1$ when we use subfield-subcodes over $\mathbb{F}_{q^2}$ from codes over $\mathbb{F}_{q^4}$ and Hermitian duality.

For RS codes, one can find a formula for $c$ in \cite{Pereira} for the Euclidean case and a partial result for the Hermitian case. We give a general formula for the Hermitian case in the next section. First, we need to introduce some notation.

\begin{de}\label{def:syasi}
We say that a minimal cyclotomic coset is symmetric if $I_x^\bot = I_{n-x} = I_{x}$ in the Euclidean case ($I_x^\bot = I_{n-qx} = I_{x}$, in the Hermitian case), and asymmetric otherwise. Let $I_x$ be asymmetric and $I_y = I_{n-x}$ ($I_y = I_{n-qx}$, in the Hermitian case). Moreover, assume that $x$ and $y$ are the minimal representatives of $I_x$ and $I_y$, respectively. Without loss of generality, we may assume that $x<y$, then we say that $I_x$ and $I_y$ are asymmetric reciprocal cosets, that $I_x$ is the first reciprocal asymmetric coset (FR-asymmetric coset) and that $I_y$ is the second reciprocal asymmetric coset (SR-asymmetric coset).
\end{de}

Let $I_R$ consists of the asymmetric cosets in $\Delta$ whose reciprocal coset does not belong to $\Delta$ and $I_L$ of the symmetric cosets in $\Delta$ and the asymmetric cosets whose reciprocal coset belongs to $\Delta$ as well. We have that $\Delta=\Delta(t) = I_R \sqcup I_L$ (that is, $\Delta=I_R \cup I_L$ and $I_R \cap I_L = \emptyset$), and in this way it holds that $E_{I_R} = E_{\Delta}  \cap E_{\Delta}^\bot$ is the hull, or radical, of $E_{\Delta}$ since $E_{\Delta}^\bot = E_{\Delta^\bot}$ (see the paragraph after Proposition \ref{la8}). Note that $I_R$ and $I_L$ are a union of minimal cyclotomic cosets. In this way, the value $c$ is given by $\# I_L$ (see Section 2 in \cite{QINP3} for more details) since$$c  =  \dim E_{\Delta} - \dim \left(E_{\Delta} \cap C_{\Delta} \right)
 =  \# \Delta - \#(\Delta \cap \Delta^\bot)=\# \Delta - \# (I_R) = \# (I_L).$$

For a cyclotomic coset $I_x$ that belongs to $I_R$, we have that $I_x \subseteq \Delta^\bot$ and, on the other hand, for a cyclotomic coset $I_x$ that belongs to $I_L$, we have that $I_x \not\subseteq \Delta^\bot$. Therefore, for obtaining regular quantum codes, one considers a set $\Delta$ such that $I_L =\emptyset$ and hence $E_{\Delta } \subseteq E_{\Delta }^\bot$.  On the contrary, for constructing LCD codes, one considers a set $\Delta$ such that $I_R = \emptyset$ and hence $E_{\Delta } \cap E_{\Delta }^\bot = \{0\}$.

\begin{rem}\label{rem:c}
Notice that since we are considering consecutive minimal cyclotomic cosets for constructing our codes, the cardinality of $I_L$ will be given by the cardinality of the symmetric cosets in  $\Delta$ plus two times the cardinality of the  SR-asymmetric cosets in $\Delta$.
\end{rem}

\begin{exa}\label{ej:3}
This is a continuation of Examples \ref{ej:1} and \ref{ej:2}, with $p=2$, $\ell =4$, $r=2$, $q^2=4$,  $n=15$, $t=6$ and $m_6=7$, where we consider the Hermitian inner product.

The symmetric cosets in $\mathbb{Z}_{15}$ are $I_0=\{0\}$, $I_5=\{5\}$, $I_{10}=\{10\}$ and the pairs of asymmetric cosets are: $I_1=\{1,4\}$ (FR-asymmetric) and $I_7=\{7,13\}$ (SR-asymmetric), $I_2=\{2,8\}$ (FR-asymmetric) and $I_{11}=\{11,14\}$ (SR-asymmetric), and $I_3 = \{3,12\}$ (FR-asymmetric) and $I_6=\{6,9\}$ (SR-asymmetric).

Thus, for $\Delta = \Delta (6) = I_0\cup I_1\cup I_2\cup I_3\cup I_5\cup I_6\cup I_7 = I_R \sqcup I_L$, we have that$$I_R = I_2, \mbox{~and~} I_L =  I_0\cup I_1\cup I_3\cup I_5\cup I_6\cup I_7,$$because $I_L= I_2$ since $I_2$ is an FR-asymmetric coset whose reciprocal coset, $I_{11}$, does not belong to $\Delta$ and because $I_L$ is a union of symmetric cosets, $I_0$ and $I_5$, and pairs of asymmetric cosets, $I_1$ and $I_7$ and $I_3$ and $I_6$ included in $\Delta$. Therefore, $c= \# I_L = 10$. Moreover, by Remark \ref{rem:c}, one can alternatively compute $c$ by considering the cardinality of the symmetric cosets, $I_1$ and $I_5$, plus two times the cardinality of the SR-asymmetric cosets, $I_6$ and $I_7$ in $\Delta$: $1+1+2(2+2)=10$. This second form of computing $c$ is the method that we follow in the rest of the article to obtain a general formula for $c$. In particular, for this code, $c$ may be computed following our forthcoming Theorem \ref{teo:gordo}, case (2), since $(q^4-1)/(q+1)=5\le m_t < q^3+q=10$.

Finally, one has that considering $\Delta=\Delta(6)$, one obtains an EAQECC with parameters $$\left[\left[n,n-2\sum_{j=0}^{t}  i_{m_j}+c,\ge m_{t+1}+1;c\right]\right]_2=\left[\left[15,1,\ge 11; 10\right]\right]_2.$$\qed
\end{exa}

Furthermore, we can also consider RS codes and BCH codes where we evaluate at zero to construct EAQECC, that is, the ideal $J$ at the beginning of Section \ref{se:dos} is generated by $X^{n+1} -X$. In this case, we obtain an EAQECC with parameters$$\left[\left[n+1,n-2\sum_{j=0}^{t}  i_{m_j}+c,\ge m_{t+1}+1;c-1\right]\right]_q,$$because the length is increased by one unit and the parameter $c$ is decreased by one unit since the cyclotomic coset $I_0=\{0\}$ is symmetric when we do not evaluate at zero and it is FR-asymmetric otherwise by \cite{QINP2} (because the reciprocal coset of $I_0 = \{0\}$ is $I_{n}$ that is not contained in $\mathcal{H}$). The dimension of the EAQECC remains therefore the same as before. This allows us to increase the constellation of codes that we may construct.

\section{Hermitian RS codes}\label{se:cuatro}

We give in this section a general formula for $c$ in the case of EAQECCs coming from RS codes with respect to the Hermitian metric.

\begin{teo}
Let $\Delta = \{ 0,1, \ldots , t\} \subseteq \mathcal{H}$ and consider the RS code $D_\Delta = E_\Delta$ over  $\mathbb{F}_{q^2}$. Let $b_0 + b_1 q$ be the $q$-adic expression of $t$, then the parameters of the corresponding EAQECC are:
$$[[q^2-1,(q-b_1)^2-2b_0-2,t +2; b_1^2+1]]_q$$ when $b_0 +b_1 < q-1$. And
$$[[q^2-1,(q-b_1-1)^2, t +2;  b_1^2 + 2(b_0 + b_1 - q ) + 4]]_q$$ otherwise ($b_0 + b_1 \ge  q-1$).
 
\end{teo}

\begin{proof}
We start by noticing that, in this case and with the notation as in Section \ref{se:dos}, every considered coset has cardinality one and $m_t =t$. By Theorem \ref{th:eaqecc}, it suffices to compute the value $c$ corresponding to the entanglement and by Remark \ref{rem:c} we need to decide which values $a$ in $\Delta$ determine a symmetric coset (contributing one to $c$) and which ones determine SR-asymmetric cosets (contributing two in the computation of the value $c$).

Let $a \in \Delta$, with $q$-adic expression $a=a_0 + a_1 q$, then $a$ represents an SR-asymmetric (respectively, a symmetric) coset if and only if $a_0 + a_1 < q-1$ (respectively, $a_0 + a_1 = q-1$). Indeed, $a$ represents an SR-asymmetric (respectively, a symmetric) coset if and only $(q^2-1) -qa <a$ (respectively, $(q^2-1) -qa =a$). Noticing that the $q$-adic expression of $q^2-1$ is $(q-1)q + (q-1)$ and that of $qa$ is $qa=a_1 + a_0 q$, the result follows straightforwardly.

Then for computing $c$ we only need to compute twice number of values (in $\Delta$) $a=a_0 + a_1 q$, $0 \leq a_0, a_1 \leq q-1$ such that $a_0+a_1 \geq q- 1$ plus the number of values (in $\Delta$)  $a=a_0 + a_1 q$, $0 \leq a_0, a_1 \leq q-1$ such that $a_0+a_1 = q- 1$.

$\Delta$ contains the number $0$ whose coset is symmetric, thus we add $1$ to our computations.

The number of elements $a$ giving rise to symmetric cosets satisfying $a_1 \leq b_1 -1$ equals to $b_1$ because each value $a_1$ determines the corresponding $a_0$.

Next, we compute how many values $a$ satisfy $a_0+a_1 \geq q- 1$ and $a_1 \leq b_1 -1$. Since for any fixed $a_1$, it holds that $a_0$ satisfies $q-1-a_1 < a_0 \leq q-1$, the number of solutions is $\sum_{a_1=0}^{b_1-1} a_1$. This sum equals $(1/2) b_1 (b_1-1)$ giving a contribution to $c$ of $b_1 (b_1-1)$.
We have proved the first part of the result because we have obtained $c= 1 + b_1 (b_1-1) + b_1 = 1 + b_1^2$. Note that we cannot obtain $a_0 + a_1  \geq q-1$ when $b_0 + b_1 < q-1$ and $a_1=b_1$.

For proving our second statement, where $b_0 + b_1 \geq q-1$, it suffices to add to the value $1 + b_1^2$ twice the number of integers $a$ with $q$-adic expression $a_0 + b_1 q$, $a_0 \leq b_0$, such that $a_0 + b_1 > q-1$ plus the number of integers $a$ with $q$-adic expression $a_0 + b_1 q$ such that $a_0 + b_1 = q-1$, which is exactly one. As a consequence,
\[
c= 1 + b_1^2 +  2(b_0 + b_1 - q +1) + 1 = b_1^2 + 2(b_0 + b_1 - q ) + 4.
\]\end{proof}

Considering the ideal $J$ introduced in Section \ref{se:dos} but now generated by $X^{n+1} -X$ instead of $X^{n} -1$, a very similar argument proves the following result.

\begin{teo}\label{th:rsq2}
Let $\Delta = \{ 0,1, \ldots , t\}  \subseteq \mathcal{H}$ and consider the RS code $D_\Delta = E_\Delta$  over $\mathbb{F}_{q^2}$. Set $t=b_0 + b_1 q$, then the parameters of the corresponding EAQECC are:
$$[[q^2,(q-b_1)^2-2b_0-2,t +2; b_1^2]]_q$$ when $b_0 +b_1 < q-1$. And
$$[[q^2,(q-b_1-1)^2,t +2; b_1^2+ 2(b_0+b_1-q)+3]]_q$$ otherwise ($b_0 + b_1 \ge  q-1$).
\end{teo}

\begin{rem}
In \cite[Theorem 5]{Pereira}, parameters for EAQECCs with length $q^2$ coming from Reed-Solomon codes were given for some intervals. Hence, our Theorem \ref{th:rsq2} extends \cite[Theorem 5]{Pereira} in the sense that we consider all possible cases.

\end{rem}

\section{Euclidean BCH EAQECC with extension degree 2}\label{se:cinco}

We consider in this section BCH codes over $\mathbb{F}_q$ with extension degree  equal to $2$. Hence, the cyclotomic cosets have one or two elements. Namely, $n= q^2-1$, $q=p^r$ and $l=2r$, and the cyclotomic coset whose minimal representative is $x$ is equal to $I_x= \{x,xq\}$. In addition, its reciprocal cyclotomic coset is $I_{n-xq} = \{n-xq,n-x\}$, where $n-xq$ is the minimal representative of $I_{n-qx}$ since $x \le xq$. Moreover, $I_x$ has cardinality $1$ if and only if $x = xq$.  

\begin{lem}\label{lem:resta}
Let $I_x$ be such that $x$ is its minimal representative. One has that 
\begin{itemize}
\item $I_x$ is an FR-asymmetric coset if and only if $n - xq > x$.
\item $I_x$ is a symmetric coset if and only if $n-xq = x$. 
\item $I_x$ is an SR-asymmetric coset if and only if $n - xq < x$.
\end{itemize}
\end{lem}

\begin{proof}
We have that $I_{n-qx}$ is the reciprocal coset of $I_x$ and that $n-xq$ is the minimal representative of $I_{n-qx}$. Hence, one has that $$I_x = I_{n-x}\mbox{~if~and~only~if~} n-xq =x.$$

Let $I_x$ be asymmetric, then $x \neq n-xq$ (otherwise it would be symmetric). Then $I_{x}$ is an SR-asymmetric cyclotomic coset if $x > n -xq$, otherwise it is an FR-asymmetric cyclotomic coset.

\end{proof}

The following proposition characterizes symmetric and SR-asymmetric cosets by the $q$-adic representation of its minimal representative.

\begin{pro}\label{pro:dec}
Let $I_x$ be such that $x$ is its minimal representative. Let $x=a_0 + a_1 q$ with $0\le a_0,a_1 < q$, the $q$-adic representation of $x$. Then
\begin{itemize}
\item $I_x$ is a symmetric coset if and only if $a_0 + a_1 = q-1$.
\item $I_x$ is an SR-asymmetric coset if and only if $a_0 + a_1 > q-1$.
\end{itemize}
\end{pro}

\begin{proof}
By Lemma \ref{lem:resta}, we have that $I_x$ is symmetric if and only if $n-xq = x$. That is, if $n - (a_0 + a_1 q)q = a_0 + a_1 q$, that is equivalent to $n = (a_0 + a_1) (q+1)$ since $q^2$ is equivalent to $1$ modulo $q^2-1$. Moreover, since $n=q^2-1$, we have that $I_x$ is a symmetric coset if $a_0 + a_1 = q-1$.

By Lemma \ref{lem:resta}, we have that $I_x$ is an SR-asymmetric coset if and only if $n-xq < x$. That is, if $n - (a_0 + a_1 q)q < a_0 + a_1 q$, that is equivalent to $n < (a_0 + a_1) (q+1)$. Thus, we have that $I_x$ is an SR-symmetric coset if and only if $a_0 + a_1 > q-1$.
\end{proof}

For computing $c$, first we compute how many minimal cyclotomic cosets are symmetric or  SR-asymmetric. Afterwards, we will determine the number of symmetric and FR-asymmetric cosets and their cardinality. The coset $I_0$ is always in $I_L$ since it is symmetric and hence, it will not be considered in the following computations. We have a characterization of the cyclotomic cosets in $I_L$ by Proposition \ref{pro:dec}, namely $a_0 + a_1 \ge q-1$ (i). Moreover, given a cyclotomic coset $I_x$ we should check whether the cyclotomic coset $I_x$ is in $\Delta=\Delta(t) = \cup_{j=0}^{t} I_{m_j}$, that is, we should check that $x \le m_t$ (ii). Finally we should consider that $x$ is a minimal representative of $I_x$, which is equivalent to $a_0 \ge a_1$ (iii), since $x = a_0 + a_1 q \le a_1 + a_0 q = xq$ if and only if $a_0 \ge a_1$.

Summarizing, setting $b_0 + b_1q$ the $q$-adic expression of $m_t$, we have to count the number of elements $a_0 + a_1q$, with $0 \le a_0, a_1 \le q-1$ (and $(a_0,a_1) \neq (0,0)$), such that 
\begin{enumerate}[(i)]
\item $a_0 + a_1 \ge q-1$,
\item $a_0 + a_1 q \le b_0 + b_1 q$, and
\item $a_0 \ge a_1$.
\end{enumerate}

We first consider the case $a_1 < b_1$. Note that in this case, condition (ii) always holds. Our strategy consists of computing the number of pairs that satisfy (i) and then subtracting the number of pairs that do not satisfy (iii). Let $a_1 =i$, with $i \in \{0 , 1, \ldots , b_1-1\}$. We have that (i) is equivalent to $q-1-i \le a_0 \le q-1$ and hence the number of possible values for $a_0$ is $i+1$ and the total number of possible values for $a_0$ and $a_1$ is\begin{equation}\label{for:ec0}\sum_{i=0}^{b_1-1} i+1 = \frac{b_1(b_1+1)}{2}.\end{equation}Now we consider pairs that do not satisfy equation (iii), that is, $a_0 < a_1$. Again, let $a_1 =i$, with $i \in \{0 , 1, \ldots , b_1-1\}$, thus we have that $q-1-i \le a_0 \le i-1$. Note that the previous inequalities are satisfied only when $(i-1) - (q-1-i) +1 = 2i -q +1 \ge 1$, that is, when $i \ge q/2$. This is equivalent to $i \ge q/2$ when the characteristic is 2 and $i \ge (q+1)/2$ when the characteristic is odd. Summarizing, both inequalities are equivalent to $i \ge \lceil q/2  \rceil$ in arbitrary characteristic. Thus, the total number of pairs that satisfy (i) and do not satisfy (iii) is equal to
$$\sum_{i=\lceil q/2  \rceil}^{b_1 -1} 2i -q +1 = \left\{ \begin{array}{l} 0, \mbox{~if~} b_1 \le \lceil q/2  \rceil, \\(b_1 - q/2 )^2, \mbox{~if~} q \mbox{~is~even~and~}b_1>\lceil q/2  \rceil,\\(b_1 - (q+1)/2)(b_1 - (q+1)/2 +1), \mbox{~if~} q \mbox{~is~odd~and~}b_1>\lceil q/2  \rceil. \end{array} \right.  $$Which, in arbitrary characteristic, equals
\begin{equation}\label{for:ec1}
\left( \max\left\{0,b_1 - \left\lceil \frac{q}{2} \right\rceil \right\}\right)\left( b_1 - \left\lfloor \frac{q}{2} \right\rfloor \right).
\end{equation}

Let us now consider the case $a_1 = b_1$. In this case, condition (i) is equivalent to $a_0 \ge q-1 -b_1$, condition (ii) is equivalent to $a_0 \le b_0$, and condition (iii) is equivalent to $a_0 \ge b_1$.  Summarizing we have that $$\max\{q-1 - b_1 , b_1 \} \le a_0 \le b_0.$$Thus, for $a_1=b_1$ the number of pairs that satisfy (i)-(iii) is \begin{equation}\label{for:ec2}\max\{0, b_0 - \max\{q-1 - b_1 , b_1 \} +1 \}.\end{equation}

As mentioned, we have counted the number of cyclotomic cosets that are  symmetric or SR-asymmetric, but we do not know yet their cardinality and how many of them are symmetric and SR-asymmetric. With respect to the number of symmetric cosets of cardinality one, we have that $x = xq$ and $n-x=x$, hence $(q^2-1)-x = x$ that implies that $x= (q^2-1)/2$. Therefore, there is a unique symmetric coset with cardinality one, $I_x$ with $x= (q^2-1)/2$, if the characteristic of $\mathbb{F}_q$ is different from 2. Otherwise, there is no symmetric coset with cardinality one. Therefore, the number of symmetric cosets of cardinality one in $\Delta$ is
\begin{equation}\label{for:ec3} \left\{ \begin{array}{l} 0, \mbox{~if~} q \mbox{~is~even}\\1, \mbox{~if~} q \mbox{~is~odd~and~} m_t\ge (q^2-1)/2. \end{array} \right. \end{equation}

We compute now the number of SR-asymmetric cosets that have cardinality one. For an asymmetric coset (FR or SR) with cardinality one, we have that $x = xq$, hence $x(q-1)$ is zero modulo $q^2-1$, and this implies that $x$ is a multiple of $(q+1)$. Therefore, the number of asymmetric cosets with cardinality one is $\lfloor (m_t)/(q+1) \rfloor$. However, we are interested in the number of SR-asymmetric cosets with cardinality one. To compute it, we should subtract the number of FR-asymmetric cosets with cardinality one. Actually, among the $q-2$ multiples of $q+1$ (we recall that we do not consider $0$) the first half are minimal representatives of FR-asymmetric cosets and the second half are the minimal representatives of SR-asymmetric cosets. Thus, if the characteristic of $\mathbb{F}_q$ is even, there are $(q-2)/2$ FR-symmetric cosets with cardinality one.  If the characteristic of $\mathbb{F}_q$ is odd, we should take into account that the coset $\{(q^2-1)/2\}$ is the only symmetric coset with cardinality one. Hence, there are $(q-3)/2 +1$ FR-asymmetric cosets with cardinality one. Therefore, for arbitrary characteristic, the number of SR-asymmetric cosets with cardinality one in $\Delta$ is \begin{equation}\label{for:ec4}\max\left\{0, \left\lfloor \frac{m_t}{q+1} \right\rfloor - \left\lfloor \frac{q-1}{2} \right\rfloor  \right\}.\end{equation}

Finally we compute the number of symmetric cosets with cardinality 2. In this case, $n - xq = x$, hence $n=x(q+1)$ and therefore $x$ is a multiple of $q-1$. Since we are not considering $0$, there are $q$ possible multiples. However, not all of them are minimal representatives, and only a half of them will  be minimal representatives of a coset. Hence, the number of symmetric cosets with cardinality 2 in $\Delta$ is \begin{equation}\label{for:ec5}\min\left\{\left\lfloor \frac{m_t}{q-1} \right\rfloor, \left\lfloor \frac{q}{2} \right\rfloor \right\}.\end{equation}

Thus, we are ready to state and prove the main result in this section.

\begin{teo}
Consider the BCH code $E_\Delta$ over the field $\mathbb{F}_q$ with extension degree 2 and length $q^2-1$ given by  $\Delta = \Delta (t) = \cup_{j=0}^t  I_{m_j} \subseteq \mathcal{H}$. Set $b_0 + b_1t$ the $q$-adic expression of $m_t$. Then the parameters of the corresponding EAQECC are 

 $$\left[\left[q^2-1,q^2-1-2\sum_{j=0}^{t}  i_{m_j}+c,\ge m_{t+1}+1;c\right]\right]_q,$$ where $c$ is equal to

\begin{dmath*}
c= 1+ 4   \left(\frac{b_1(b_1+1)}{2}-\left( \max\left\{0,b_1 - \left\lceil \frac{q}{2} \right\rceil \right\}\right)\left( b_1 - \left\lfloor \frac{q}{2} \right\rfloor \right)   \\ \hiderel{+} \max\{0, b_0 - \max\{q-1 - b_1 , b_1 \} +1 \} \right) \\ \hiderel{-}3\delta -2 \left( \max\left\{0, \left\lfloor \frac{m_t}{q+1} \right\rfloor - \left\lfloor \frac{q-1}{2} \right\rfloor  \right\} \right) -2 \left( \min\left\{\left\lfloor \frac{m_t}{q-1} \right\rfloor, \left\lfloor \frac{q}{2} \right\rfloor \right\}\right),
\end{dmath*}and where $\delta$ is equal to $1$ if $q$  is odd and $m_t\ge (q^2-1)/2$, and it is equal to $0$ otherwise.





\end{teo}

\begin{proof}
By Remark \ref{rem:c}, we have that $c$ is equal to the sum of the cardinality of the symmetric cosets in $\Delta$ plus two times the cardinality of the SR-asymmetric cosets in $\Delta$.

The coset $I_0 = \{0\}$ is symmetric and it always contributes with 1 to the value of $c$. Consider the above referenced values from (\ref{for:ec0}) to (\ref{for:ec5}). Then (\ref{for:ec0})$-$(\ref{for:ec1})$+$(\ref{for:ec2}) are the number of symmetric and SR-asymmetric cosets in $\Delta$. If they were all SR-asymmetric cosets with cardinality 2, they would contribute with $4$ times (\ref{for:ec0})$-$(\ref{for:ec1})$+$(\ref{for:ec2}) to the value $c$. However, this may not be the case and we should adjust the previous computation.

Note that a coset with cardinality one contributes 1 to $c$ if it is symmetric and it contributes 2 to $c$ if it is SR-asymmetric. Finally, a coset with cardinality 2 that is symmetric contributes 2 to $c$.

Therefore,  $$c= 1+ 4 \times ((\ref{for:ec0})-(\ref{for:ec1})+(\ref{for:ec2})) -3\times(\ref{for:ec3})-2\times(\ref{for:ec4})-2\times (\ref{for:ec5}),$$and the result holds.
\end{proof}

If one constructs the above codes as described in Section \ref{se:dos} but with the ideal $J$ generated by $X^{n+1} -X$ instead of $X^{n} -1$, a very similar argument proves the following result.

\begin{teo}
Consider the BCH code $E_\Delta$ over the field $\mathbb{F}_q$ with extension degree 2 and length $q^2$ given by  $\Delta = \Delta (t) = \cup_{j=0}^t  I_{m_j} \subseteq \mathcal{H}$. Set $b_0 + b_1t$ the $q$-adic expression of $m_t$. Then the parameters of the corresponding EAQECC are 

 $$\left[\left[q^2,q^2-2\sum_{j=0}^{t}  i_{m_j} +c,\ge m_{t+1}+1;c\right]\right]_q,$$ where $c$ is equal to

\begin{dmath*}
c= 4   \left(\frac{b_1(b_1+1)}{2}-\left( \max\left\{0,b_1 - \left\lceil \frac{q}{2} \right\rceil \right\}\right)\left( b_1 - \left\lfloor \frac{q}{2} \right\rfloor \right)   \\ \hiderel{+} \max\{0, b_0 - \max\{q-1 - b_1 , b_1 \} +1 \} \right) \\ \hiderel{-}3\delta -2 \left( \max\left\{0, \left\lfloor \frac{m_t}{q+1} \right\rfloor - \left\lfloor \frac{q-1}{2} \right\rfloor  \right\} \right) -2 \left( \min\left\{\left\lfloor \frac{m_t}{q-1} \right\rfloor, \left\lfloor \frac{q}{2} \right\rfloor \right\}\right),
\end{dmath*}and where $\delta$ is equal to $1$ if $q$  is odd and $m_t\ge (q^2-1)/2$, and it is equal to $0$ otherwise.

\end{teo}

Note that considering an RS code over the finite field with $q^2$ elements and the Hermitian metric produces an EAQECC with better parameters than the corresponding one obtained with a BCH over the field with $q$ elements and the Euclidean metric, because both have the same length and bound for the minimum distance and the dimension is larger for the RS code with the Hermitian metric. However, we notice that it is possible to obtain a BCH code with a higher value for $c$, which increases the constellation of known EAQECCs.

\section{Hermitian BCH EAQECC  with extension degree 2}\label{se:seis}

We consider now BCH codes over $\mathbb{F}_{q^2}$ and the Hermitian inner product to construct EAQECCs. Set $\ell = 2r$ and $n= q^4 -1$ in this section. Thus, the cyclotomic coset with minimal representative $x$ is equal to $I_x = \{x , q^2 x \}$, and the reciprocal cyclotomic coset of $I_x$ is equal to $I_{n-qx}$.

We can again  characterize symmetric cosets and SR-asymmetric cosets in $\Delta=\Delta(t) = \cup_{j=0}^{t} I_{m_j}$, but such a characterization is not as precise as for the Euclidean case for cosets with cardinality 2. Let $x$ be the minimal representative of $I_x$ with cardinality 2. In the Euclidean case, we have that $n-qx$ is  the minimal representative of its reciprocal coset $I_{n-x}$. However, in the Hermitian case the reciprocal coset of $I_x$ is $I_{n-qx}= \{n-qx, n-q^3 x \}$ and we do not know a priori who is its minimal representative because any one of them can be the smallest element modulo $q^4-1$.

Let $(a_0,a_1,a_2,a_3)$ be the $q$-adic expansion of $x$, i.e., $x=a_0+a_1q+a_2q^2+a_3 q^3$, with $0 \le a_i < q$ for $0 \le i \le 3$. We denote in this section by $x$ the 4-tuple $(a_0,a_1,a_2,a_3)$ as well. Note that $y=q^2x$ has $q$-adic expansion $(a_2,a_3,a_0,a_1)$. We begin by studying the cosets with cardinality one since they can be easily characterized.

\begin{lem}\label{lem:card1}
$I_{x} = \{x\}$ holds if and only if $(q^2 +1) \mid x$. Moreover, $I_x$ has cardinality 1 if and only if the $q$-adic expansion of $x$ is of the form $(a,b,a,b)$, with $0\le a,b <q$.
\end{lem}

\begin{proof}
$I_{x}=\{x\}$ if and only if $x=q^2 x$. That is, $x(q^2-1) =0$ modulo $n=q^4-1$. Let  $x=(a_0,a_1,a_2,a_3)$, since $I_x=\{x\}$ we have that $x= q^2x$ and $q^2(a_0,a_1,a_2,a_3)=(a_2,a_3,a_0,a_1)$. Therefore $(a_0,a_1,a_2,a_3)=(a_2,a_3,a_0,a_1)$ and the result holds.
\end{proof}

Next, we characterize symmetric and SR-asymmetric cosets with cardinality one.

\begin{lem}\label{le:CuandoLCDunico}
Let $x = (a_0,a_1,a_2,a_3)$ with $I_x=\{x\}$, then
\begin{itemize}
\item $I_x$ is a symmetric coset if and only if $a_2 + a_3 = q-1$.
\item $I_x$ is an SR-asymmetric coset if and only if $a_2 + a_3 > q-1$.
\end{itemize}
\end{lem}

\begin{proof}
Since $I_x= \{x\}$, one has that $I_x$ is symmetric if $n-qx=x$, that is, if $$(q-1,q-1,q-1,q-1)-(a_3,a_2,a_3,a_2) = (a_2,a_3,a_2,a_3),$$which is equivalent to $a_2 + a_3 = q-1$. Analogously, $I_x$ is an SR-asymmetric coset if $n-qx<x$, that is, if $$(q-1,q-1,q-1,q-1)-(a_3,a_2,a_3,a_2) < (a_2,a_3,a_2,a_3),$$which is equivalent to $a_2 + a_3 > q-1$.
\end{proof}

Thus, we have completely characterized the cosets with cardinality one. We consider now cosets with cardinality two. As we have mentioned before, its study will be more elaborated and complicated, although we have some positive news described in the following result.

\begin{lem}\label{lem:bombazo}
There is no symmetric coset with cardinality 2 in the Hermitian case.
\end{lem} 

\begin{proof}
Let $I_x$ be with cardinality 2 and symmetric, then $n-qx = q^2 x$ or $n -qx = x$. Let us show, by contradiction, that these equalities do not hold. If $n-qx = q^2 x$, then $q^4-1 = q(q-1)x$ which implies that $q$ divides $q^4-1$, contradiction. 

If $n-qx = x$, then $n = (q+1)x$. This implies that $x= (q-1)(q^2 +1) = -1 +q -q^2 + q^3 = (q-1,0,q-1,0)$ and, by Lemma \ref{lem:card1}, $I_x$ has cardinality 1, contradiction.
\end{proof}

Hence, among the cosets with cardinality two, we should only count SR-asymmetric cosets. 

\begin{lem}
Let $I_x$ be with cardinality 2 and such that $x$ is its minimal representative. Then $I_x$ is an SR-asymmetric coset if and only if there exists $z \in I_x$ such that $n-qz < x$.
\end{lem} 

\begin{proof}
Let $I_x = \{x, y\}$. Note that the minimal representative of the reciprocal cyclotomic coset of $I_x$ is $n-qx$ or $n-qy$ and the result follows as in Lemma \ref{lem:resta}.
\end{proof}

  The following result is used to characterize  which element is the minimal representative of the reciprocal coset of a coset with two elements.

\begin{lem}\label{le:QuienMayor}
Let $I_x=\{x,y\}$ be with cardinality 2 and such that $x=(a_0,a_1,a_2,a_3)$ is its minimal representative. One has that:
\begin{itemize}
\item $qx>qy$ if and only if $a_2>a_0$ or $a_2=a_0$ and $a_1>a_3$.
\item $qx<qy$ if and only if $a_2<a_0$ or $a_2=a_0$ and $a_1<a_3$.
\end{itemize}
\end{lem}

\begin{proof}
We have that $qx=(a_3,a_0,a_1,a_2)$ and $qy=(a_1,a_2,a_3,a_0)$. We prove the case $	qx>qy$, because the case $qx<qy$ is analogous. The inequality $qx>qy$ holds if and only if $(a_3,a_0,a_1,a_2)$ is greater than $(a_1,a_2,a_3,a_0)$ with respect to the lexicographical ordering. That is, if $a_2 >a_0$ or if $a_2 = a_0$ and $a_1 > a_3$. Note that since $qx \neq qy$, we cannot have that  $a_2 = a_0$ and $a_1 = a_3$ and the result holds.
\end{proof}

Next we characterize the cyclotomic cosets with two elements that are SR-asymmetric cosets.

\begin{lem}\label{le:CuandoLCD}
Let $I_x=\{x,y\}$ with cardinality 2, where $x=(a_0,a_1,a_2,a_3)$ is its minimal representative. Then, $I_x$ is an SR-asymmetric coset if and only if 

\begin{itemize}
\item either $a_2+a_3 > q-1$, or $a_2+a_3=q-1$ and $a_1+a_2>q-1$, if $qx>qy$,
\item either $a_0+a_3 > q-1$, or $a_0+a_3=q-1$ and $a_3+a_2>q-1$, if $qx<qy$.
\end{itemize}
\end{lem}

\begin{proof}

Let us assume  that $qx >qy$, the case $qx < qy$ follows analogously. In this case, the minimal representative of $I_{n-qx}$ is $n-qx$. We should compute when $n-qx \le x$, that is, when $(q-1,q-1,q-1,q-1)-(a_3,a_0,a_1,a_2) \le (a_0,a_1,a_2,a_3)$, since $n=(q-1,q-1,q-1,q-1)$. That is, \begin{equation}\label{eq:lem4}
(q-1,q-1,q-1,q-1) \le (a_3 + a_0)+ (a_0 + a_1)q + (a_1 + a_2)q^2 + (a_2+ a_3)q^3.
\end{equation} We claim that (\ref{eq:lem4}) holds if $a_2 + a_3 > q-1$, or if $a_2 + a_3 = q-1$ and $a_1 + a_2 > q-1$. Equivalently, we are claiming that if $a_2 + a_3 = q-1$ and $a_1 + a_2 = q-1$, then Inequality (\ref{eq:lem4}) does not hold. Indeed, in that case we would have that $a_1 = a_3$ and, by Lemma \ref{le:QuienMayor},  $a_2 > a_0$ (since $qx > qy$). As a consequence, (\ref{eq:lem4}) is not true since $a_0 + a_1 < a_2 + a_1 = q-1$.\end{proof}


The following definition is a key concept to study cosets with cardinality 2.

\begin{de}
We define an {\it interlude} as the set of natural numbers that are between two consecutive cosets of cardinality one and that are minimal representatives of a coset with cardinality 2. By Lemma \ref{lem:card1}, the cyclotomic cosets of cardinality 1 are of the form $\{(a,b,a,b)\}$, with $0\le a,b <q$. Hence, an interlude is formed by integers $x$ of the form $$(a,b,a,b)<x<(a+1,b,a+1,b), \mbox{~with~} 0 \le a <q-1\mbox{~and~} 0\le b <q,$$ or  $$(q-1,b,q-1,b)<x<(0,b+1,0,b+1) \mbox{~with~} 0\le b \le q-2, \mbox{~and~} x<q^2 x,$$ to ensure that $x$ is the minimal representative. We denote the interlude bounded by $(a,b,a,b)$ and $(a+1,b,a+1,b)$ by $[(a,b,a,b),(a+1,b,a+1,b)]_M$.
\end{de}

We now characterize the elements in an interlude. We will see that they are all consecutive which explains the name.

\begin{lem}\label{le: CuantasXentre2}
Let $0 \le a <q-1$ and $0\le b <q$, then $[(a,b,a,b),(a+1,b,a+1,b)]_M =$ $$ \{ (i,j,a,b) : 0\le i < q \mbox{~and~} b < j <q, \mbox{~or~}  a < i <q \mbox{~and~} j=b\}.$$

Let $0\le b \le q-2$, then $$[(q-1,b,q-1,b), (0,b+1,0,b+1)]_M = \{ (i,j,q-1,b) : 0\le i < q \mbox{~and~} b < j <q \}.$$
\end{lem}
\begin{proof}

Let us consider $(a,b,a,b)<x<(a+1,b,a+1,b)$, then $x$ can be written in two different ways:
\begin{enumerate}[(i)]
\item $x=(i,j,a,b)$, with $0 \le i <q$ and $b <j <q$, or $a < i <q$ and $j=b$.
\item $x=(i,j,a+1,b)$, with $0 \le i <q$ and $0 \le j < b$, or $0 \le i < a+1$ and $j=b$.
\end{enumerate}
We claim that all elements of the form described in (i) are minimal representatives, and thus they are in an interlude, and that all the elements in (ii) are not minimal representatives, and hence they are not in an interlude. Then $x=(i,j,a,b)$ is the minimal representative of its cyclotomic coset if $x<q^2x$ and then  $q^2(i,j,a,b)=(a,b,i,j)>(i,j,a,b)$ if and only if $b<j$ or $j=b$ and $a<i$. That is, we get all the elements in (i). Consider now $x=(i,j,a+1,b)$, again $x$ is the minimal representative of its coset if $q^2 (i,j,a+1,b) = (a+1,b,i,j) > (i,j,a+1,b)$, which implies that $b < j$, or that $a+1 < i$ and $b=j$. Note that these conditions are not satisfied by any element in (ii).

We consider now  $(q-1,b,q-1,b)<x<(0,b+1,0,b+1)$, then $x$ can be written in two different ways:
\begin{enumerate}[(i)]
\setcounter{enumi}{2}
\item $x=(i,j,q-1,b)$, with $0 \le i <q$ and $b <j <q$.
\item $x=(i,j,0,b+1)$, with $0 \le i <q$ and $0 \le j < b +1$.
\end{enumerate}Again, we claim that the elements in (iii) are in an interlude and all the elements in (iv) are not. One has that $x=(i,j,q-1,b)$ is the minimal representative of its cyclotomic coset if $x<q^2x$ and then  $q^2(i,j,q-1,b)=(q-1,b,i,j)>(i,j,a,b)$ if and only if $b<j$. That is, we obtain all the elements in (iii). Consider now $x=(i,j,0,b+1)$, again $x$ is the minimal representative of its coset if $q^2 (i,j,0,b+1) = (0,b+1,i,j) > (i,j,0,b+1)$, which implies that $b +1 < j$. This condition is not satisfied by any element in (iv), which concludes the proof.
\end{proof}

\begin{rem}\label{rem:xintervalo}
Let $I_x$ be with cardinality 2 and $x=(a_0,a_1,a_2,a_3)$ its minimal representative, then $x$ is in the interlude $[(a_2,a_3,a_2,a_3),(a_2+1,a_3,a_2+1,a_3)]_M$ if $a_2 < q-1$, and in $[(q-1,a_3,q-1,a_3),(0,a_3+1,0,a_3+1)]_M$ if $a_2 = q-1$.
\end{rem}

FR-asymmetric cosets were described in \cite[Theorem 3.11]{QINP2}, where they were used to construct regular quantum codes:

\begin{lem}\label{le:quantumpart}
Let $0< x< (q^4-1)/(q+1)$, then $I_x$ is an FR-asymmetric coset and $x$ is its minimal representative.  
\end{lem}

Next we characterize the $SR$-asymmetric cyclotomic cosets whose minimal representative is in the interlude $[(q^4-1)/(q+1),q^3+q]_M$. 


\begin{lem}\label{le:Hasta_q3+q}
Let $x$ be the minimal representative of $I_x$, a coset with cardinality 2, with $(q^4-1)/(q+1)< x < q^3+q$. Then $I_x$ is an SR-asymmetric coset. Furthermore, there are $q^2-q$ minimal representatives of SR-asymmetric cosets in the interlude $[(q^4-1)/(q+1),q^3+q]_M$
\end{lem}

\begin{proof}
Note that $(q^4-1)/(q+1)=(q-1,0,q-1,0)$ and $q^3+q=(0,1,0,1)$. By Lemma \ref{le: CuantasXentre2}, an element $x$ in $[(q^4-1)/(q+1),q^3+q]_M$ is given by $x=(a_0,a_1,a_2,a_3)=(i,j,q-1,0)$, with $0 \le i < q$ and $0 < j < q$. Hence, there are $q(q-1)$ elements in the interlude that are minimal representatives of cosets, we claim that all these cosets are SR-asymmetric.

In fact, $qx=(0,i,j,q-1)$ and $qy=(j,q-1,0,i)$, where $I_x = \{x,y\}$. Thus,  $qx > qy$ and by Lemma \ref{le:QuienMayor},  it holds that $a_2>a_0$ or $a_2=a_0$ and $a_1>a_3$. By Lemma \ref{le:CuandoLCD}, the result follows because $a_2 + a_3 = q-1$ and $a_1 + a_2 > q-1$ since $(q-1) + 0 = (q-1)$ and $j + (q-1) > q-1$ (because $j > 0$).
\end{proof}

Our aim in the following lemmas is to count the quantity of  SR-asymmetric cosets in an arbitrary interlude. We divide our study in two cases, for $a_2 + a_3 < q-1$ and $a_2 + a_3 \ge q-1$. These results will be used when proving the main theorem.
 
 
\begin{lem}\label{le:a2+a3<q-1}
There are $a_3 (q-a_3)$ minimal representatives of SR-asymmetric cosets with cardinality 2 in the interlude $[(a_2,a_3,a_2,a_3),(a_2+1,a_2,a_2+1,a_3)]_M$, with $a_2 + a_3 < q-1$.
\end{lem}
\begin{proof}
Let $x \in [(a_2,a_3,a_2,a_3),(a_2+1,a_2,a_2+1,a_3)]_M$. By Lemma \ref{le: CuantasXentre2}, $x=(i,j,a_2,a_3)$, with  $0\le i < q$ and $a_3 < j <q$, or $ a_2 < i <q$   and $j=a_3$. Since $a_2 + a_3 < q-1$, by Lemma \ref{le:CuandoLCD} and Lemma \ref{lem:bombazo}, we have that $I_x$ is an SR-asymmetric coset with cardinality two if and only if $qx < qy$. Therefore, we have that $qx<qy$. By Lemma \ref{le:CuandoLCD}, $I_x$ is an SR-asymmetric coset if $i+a_3>q-1$ or if $i+a_3=q-1$ and $a_2+a_3>q-1$. Therefore, $I_x$ is an SR-asymmetric coset if $i + a_3 > q-1$ because we are assuming that $a_2+a_3 < q-1$.

Let $a_3< j< q$, then there are $a_3$ possible values for those $i$ such that $i+a_3 > q-1$, since $i<q$, namely $q-1-a_3 + 1, q-1-a_3+2, \ldots, q-1$. We consider now the case $j = a_3$, then there are also $a_3$ values $i$ satisfying $i+a_3 > q-1$, but we have a stronger restriction: $a_2 <i <q$ in this case, however, the assumption $a_2+a_3< q-1 $ implies that all possible values of $i$ satisfy the restriction $a_2 <i <q$ as well. Summarizing, $j$ can take $q-a_3$ possible values and, for all of them, $i$ may take $a_3$ values. Hence, there are $(q-a_3)a_3$ SR-asymmetric cosets with cardinality 2 in the given interlude.
\end{proof}

\begin{lem}\label{le:a2+a3>=q-1}
There are $q^2 - qa_3 - a_2 -1$ minimal representatives of SR-asymmetric cosets with cardinality 2 in the interlude $[(a_2,a_3,a_2,a_3),(a_2+1,a_2,a_2+1,a_3)]_M$, with $a_2 + a_3 \ge q-1$, and in the interlude $[(q-1,a_3,q-1,a_3),(0,a_3+1,0,a_3+1)]_M$ (considering $a_2=q-1$). Actually, every cyclotomic coset generated by an element in these interludes is an SR-asymmetric coset.
\end{lem}
\begin{proof}

We start by proving our first statement. Let $x \in [(a_2,a_3,a_2,a_3),(a_2+1,a_2,a_2+1,a_3)]_M$. By Lemma \ref{le: CuantasXentre2}, $x=(i,j,a_2,a_3)$, with  $0\le i < q$ and $a_3 < j <q$, or $a_2 < i <q$   and $j=a_3$. We consider first the case $qx>qy$. If $a_2 + a_3 > q-1$ then $I_x$ is an SR-asymmetric coset  by Lemma \ref{le:CuandoLCD}. If $a_2 + a_3 = q-1$, $I_x$ is an SR-asymmetric coset if $j + a_2 > q-1$ by Lemma \ref{le:CuandoLCD}. Note that $j + a_2 > q-1$ if and only if $j > a_3$,  by Lemma \ref{le: CuantasXentre2}, since we are assuming that $a_2 + a_3 = q-1$ in this case. We consider now the case $qx<qy$. By Lemma \ref{le:QuienMayor}, $i >a_2$ since $j \ge a_3$. Moreover, since we assume that $a_2 + a_3 \ge q-1$ and $i > a_2$, one has that $i + a_3 > q-1$ and by Lemma \ref{le:CuandoLCD}, $I_x$ is an SR-asymmetric coset. Finally, we determine how many minimal representatives of SR-asymmetric cosets are. For every $j$ such that $a_3 < j < q$ we have $q$ possible values for $i$ (since $0\le i < q$), hence we have $(q-1-a_3)q$ minimal representatives. For $a_3=j$, we have $q-a_2-1$ minimal representatives and then the first statement of this result holds. 

We prove now the second statement. Assume that $x$ is in $[(q-1,a_3,q-1,a_3),(0,a_3+1,0,a_3+1)]_M$. By Lemma \ref{le: CuantasXentre2}, $x=(a_0,a_1,a_2,a_3)=(i,j,q-1,a_3)$, with  $0\le i < q$ and $a_3 < j <q$, and hence there are $q(q-1-a_3)$  possible values for $i$ and $j$. It remains to prove that $I_x$ is an SR-asymmetric coset for all $x$. The case $a_3=0$ follows from Lemma \ref{le:Hasta_q3+q}, hence we can assume that $a_3>0$. Moreover, we have that $qx>qy$ by Lemma \ref{le:QuienMayor}, since $a_2 = q-1 \le a_0$ and $a_1 = j > a_3$. Therefore, $I_x$ is an SR-asymmetric coset by Lemma \ref{le:CuandoLCD}, because $a_2 + a_3 = q-1 + a_3 > q-1$ (since $a_3 >0$).
\end{proof}

We can now prove the main result that gives the parameters of the EAQECCs we construct in the Hermitian case. 

\begin{teo}\label{teo:gordo}
Consider the BCH code $E_\Delta$ over the field $\mathbb{F}_{q^2}$ with extension degree 2 and length $q^4-1$, where  $\Delta = \Delta (t) = \cup_{j=0}^t  I_{m_j} \subseteq \mathcal{H}$. Set $(b_0,b_1,b_2,b_3)$ the $q$-adic expression of $m_t$. Then, considering the Hermitian inner product, the parameters of the corresponding EAQECC are 

 $$\left[\left[q^4-1,q^4-1-2\sum_{j=0}^{t}  i_{m_j}+c,\ge m_{t+1}+1;c\right]\right]_q,$$ where

\begin{enumerate}
\item If $m_t<\frac{q^4-1}{q+1}$, then $c=1$.
\item If $\frac{q^4-1}{q+1}\le m_t <q^3+q$, then $c=2+4\left(m_t - \frac{q^4-1}{q+1}\right)$.

\item If $q^3+q\le m_t$, then

\begin{dmath*}
c=1 +  4q(q-1)    + b_3^2  + \max\{0,2(b_2+b_3-q+2)-1\}   \\ \hiderel{+} 4\left( \sum_{j=1}^{b_3-1} \left[ (q-1-j)(q-j)j + (j+1)(q^2 - qj -1) - \frac{(2(q-1)-j)(j+1)}{2} \right]   \\ \hiderel{+}
\delta\left(b_2b_3(q-b_3)+ b_3(b_1-b_3)+\max\{0,b_0+b_3-q+1)\}\right)   \\ \hiderel{+}  (1-\delta)\left[ (q-1-b_3)b_3(q-b_3) + (b_2 + b_3 -q +1)((q^2-q b_3 -1) -   (b_2 + q -b_3 -2)/2) + (m_t- (b_2+b_3q+b_2q^3+b_3q^3)) \right] \right) ,\end{dmath*}$\delta$ being 1 if $b_2 + b_3 <q-1$ and $\delta = 0$ otherwise.



 

\end{enumerate}
\end{teo}

\begin{proof}

By Remark \ref{rem:c}, we have that $c$ is equal to the cardinality of the symmetric cosets in $\Delta$ plus two times the cardinality of the SR-asymmetric cosets in $\Delta$. Note that the symmetric cosets of cardinality 1 contribute 1 to $c$ and the SR-asymmetric cosets of cardinality 1 contribute 2 to $c$. The SR-asymmetric cosets of cardinality 2 contribute 4 to $c$ and that there are no symmetric cosets with cardinality 2 by Lemma \ref{lem:bombazo}.

We have that the coset $I_0 = \{0\}$ is symmetric and it always contributes 1 to the value of $c$.

{\it (1)} Let $m_t<\frac{q^4-1}{q+1}$, then the results holds by Lemma \ref{le:quantumpart}.

{\it (2)} Let $(q^4-1)/(q+1)\le m_t<q^3+q=(q^4-1)/(q+1)+q^2 +1$. By Lemma \ref{lem:card1}, the cyclotomic coset generated by $(q^4-1)/(q+1)$ has cardinality 1 and it is symmetric by Lemma \ref{le:CuandoLCDunico}.

By Lemma \ref{le:Hasta_q3+q}, an element greater than $(q^4-1)/(q+1)$ and smaller than or equal to $m_t$, that is a minimal representative of a cyclotomic coset, generates an SR-asymmetric coset. Hence there are $m_t - (q^4-1)/(q+1)$ SR-asymmetric cosets. 


{\it (3)} Let $q^3+q \le m_t$, by Lemma \ref{le:Hasta_q3+q} we have $q(q-1)$ minimal representatives of SR-asymmetric cosets in the interlude $[(q^4-1)/(q+1),q^3+q]_M$, that together with the symmetric coset with one element of $(q^4-1)/(q+1)$ contribute $$4q(q-1) +1$$ to $c$.

We count first the symmetric and SR-asymmetric cosets with cardinality 1, that is $I_x=\{x\}$ with $(i,j,i,j)=x\le m_t=(b_0,b_1,b_2,b_3)$. Note that $(q-1,0,q-1,0)$ has been already considered in part (2). By Lemma \ref{le:CuandoLCDunico}, we have that $I_x$ is symmetric if $i+j = q-1$ and that $I_x$ is an SR-asymmetric coset if $i+ j > q-1$. 
\begin{itemize}
\item Let $1 \le j < b_3$, we have that $I_x$ is symmetric for $i=q-1-j$ and that is an SR-asymmetric coset for $ q- j  \le i < q$. Hence, for $1 \le j < b_3$, we have $b_3 -1 $ symmetric cosets and $\sum_{j=1}^{b_3-1} j = b_3(b_3 -1)/2$ SR-asymmetric cosets. 
\item For the case $j = b_3$, we should consider $(i,b_3,i,b_3) \le (b_0,b_1,b_2,b_3)$ that is equivalent to $i \le b_2$ by Lemma \ref{le: CuantasXentre2}. Moreover, one has that $i+b_3 \ge q-1$. Hence, there are $\max\{0,b_2+b_3-q+2\}$ possible values for $i$. Only one of them generates a symmetric coset and the others are SR-asymmetric cosets. 

\end{itemize}

Hence, the cosets with cardinality one, excepting $(q^4-1)/(q+1)$, contribute  $$b_3 -1 + b_3(b_3 -1) + \max\{0,2(b_2+b_3-q+2)-1\} = $$ $$  -1 + b_3^2 + \max\{0,2(b_2+b_3-q+2)-1\} $$ to $c$. 

Now, we focus on the rest of the SR-asymmetric cosets with cardinality 2. We first count the number of different interludes that are smaller than the interlude that contains $m_t$ and then count the number of SR-asymmetric cosets in each interlude. By Remark \ref{rem:xintervalo}, $m_t=(b_0,b_1,b_2,b_3)$ is in the interlude $[(b_2,b_3,b_2,b_3),(b_2+1,b_3,b_2+1,b_3)]_M$ if $b_2 < q-1$ and in $[(q-1,b_3,q-1,b_3),(0,b_3+1,0,b_3+1)]_M$ if $b_2 = q-1$. Thus, the interludes that we should count are of the form $[(i,j,i,j),(i+1,j,i+1,j)]_M$  or $[(q-1,j,q-1,j),(0,j+1,0,j+1)]_M$  with   $0 \le i < q$ and $0< j < b_3$, or $ 0 \le i  < b_2$  and $j = b_3$. We divide the study in two cases $j < b_3$ and $j=b_3$.
 
Let $j < b_3$. We divide again the study in two cases, the first one where we can use Lemma \ref{le:a2+a3<q-1} and the second one where we can use Lemma \ref{le:a2+a3>=q-1}.
\begin{itemize}
\item There are  $q-1-j$ interludes of type $[(i,j,i,j),(i+1,j,i+1,j)]_M$, with $i+j<q-1$. By Lemma \ref{le:a2+a3<q-1}, they contain $(q-j)j$ SR-asymmetric cosets. Hence the case $j < b_3$ and $i+j<q-1$ contributes $$4\sum_{j=1}^{b_3 -1} (q-1-j)(q-j)j$$ to $c$.


\item We consider now the opposite case,  $i+j\ge q-1$, where we know that every minimal representative in the interlude is an SR-asymmetric coset by Lemma \ref{le:a2+a3>=q-1}. Namely, there are $q^2 - qj -i -1$ elements in the interlude with first element $(i,j,i,j)$ with $0<j<b_3$ and $q-1-j \le i < q$. Therefore, we have in total $$\sum_{i=q-1-j}^{q-1} (q^2 - qj -i -1 )= (j+1)(q^2 - qj -1) - \sum_{i=q-1-j}^{q-1} i,$$which means that the case 
$j < b_3$ and  $i+j\ge q-1$ contributes $$4\left((j+1)(q^2 - qj -1) - \frac{(2(q-1)-j)(j+1)}{2}\right)$$   to $c$.
\end{itemize} 
 
Let $j = b_3$. As above, we divide the study in two cases. 


\begin{itemize}
\item We consider first the case $b_2+b_3<q-1$. Note that there are $b_2$ interludes of the form $[(i,j,i,j),(i+1,j,i+1,j)]_M$ with $0 \le i <b_2$ and $j=b_3$ that are smaller than $(b_2,b_3,b_2,b_3)$. Moreover, by Lemma \ref{le:a2+a3<q-1}, they contain $b_3(q-b_3)$ minimal representatives. 

Furthermore, we should count how many elements in the interlude containing $m_t$, $[(b_2,b_3,b_2,b_3), (b_2+1,b_3,b_2+1,b_3)]_M$ (recall Remark \ref{rem:xintervalo}), are smaller than or equal to $m_t$. In this interlude the elements are $(i,k,b_2,b_3)$ with $0 \le i < q$ and $b_3 <k<q$, or $b_2 <i <q$ and $k=b_3$. 

Those elements that are minimal representatives of an SR-asymmetric coset satisfy $i+ b_3 > q-1$ by Lemma \ref{le:CuandoLCD}, because $b_2+b_3<q-1$ (we have that $qx >qy$ does not hold if $b_2 + b_3 <q-1$). Hence, combining both restrictions we have either $q-b_3 \le i <	q$ and $b_3 <k<q$ or $\max\{b_2 + 1, q-b_3\} \le i <q$. Note that under the hypothesis that $b_2 + b_3 < q-1$, we have that $\max\{b_2 + 1, q-b_3\} = q-b_3$ holds. Hence, the restrictions for belonging to the interlude and generating an SR-asymmetric coset are either$$q-b_3 \le i < q \mbox{~and~} b_3 <k<q$$or$$q-b_3 \le i <q.$$

Moreover, these elements are smaller than or equal to $m_t$, that is, $k < b_1$ or $k=b_1$ and $i \le b_0$. From these restrictions we get (note that $b_3 \le b_1$ by Lemma \ref{le: CuantasXentre2}): 

$$q-b_3 \le i <q \mbox{~and~} b_3 <k <b_1,$$ or $$q-b_3 \le i \le b_0 \mbox{~and~} k=b_1,$$  or $$q-b_3 \le i <q \mbox{~and~} k=b_3.$$

Thus the number of possible representatives of an SR-asymmetric coset is:

$$b_3(b_1-b_3-1) + \max\{0,b_0+b_3-q+1\} + b_3  =$$ $$ b_3(b_1 -b_3)  + \max\{0,b_0+b_3-q+1\}$$
Hence, this case contributes $$4(b_2b_3(q-b_3)+ b_3(b_1-b_3)+\max\{0,b_0+b_3-q+1)\})$$ to $c$.

\item We consider now the case $b_2+b_3\ge q-1$. Note that there are $b_2$ interludes of the form $[(i,j,i,j),(i+1,j,i+1,j)]_M$ with $0 \le i <b_2$ and $j=b_3$ whose elements are smaller than $(b_2,b_3,b_2,b_3)$. For $0 \le i  < q-1 -b_3$, we have that $i + b_3 < q-1$, and by Lemma \ref{le:a2+a3<q-1}, they contain $b_3(q-b_3)$ minimal representatives, which means that the first $q-1-b_3$ interludes have a total of  $(q-1-b_3)b_3(q-b_3)$ minimal representatives of SR-asymmetric cosets. 

For the remaining cosets, $q-1-b_3 \le i < b_2$, we have that $i + b_3 \ge q-1$ holds and thus we should consider Lemma \ref{le:a2+a3>=q-1}. Every minimal representative of these interludes generates an SR-asymmetric coset and there are $q^2-q b_3 -i-1$.  Therefore, the number of minimal representatives in this case is

$$\sum_{i=q-1-b_3}^{b_2 -1} (q^2 -qb_3 -i -1) = $$

$$(b_2 + b_3 -q +1)(q^2-q b_3 -1) -  \sum_{i=q-1-b_3}^{b_2 -1} i = $$

$$(b_2 + b_3 -q +1)(q^2-q b_3 -1) -  (b_2 + q -b_3 -2)(b_2+b_3-q+1)/2.$$





Furthermore, we should count how many elements in the interlude that contains $m_t$, $[(b_2,b_3,b_2,b_3), (b_2+1,b_3,b_2+1,b_3)]_M$, are smaller than or equal to $m_t$. Since, by Lemma \ref{le:a2+a3>=q-1}, all elements in the interlude of $m_t$ are minimal representatives of an SR-asymmetric coset, we conclude that there are $m_t- (b_2+b_3q+b_2q^3+b_3q^3)$ elements.

Hence, this case contributes \begin{dmath*}(q-1-b_3)b_3(q-b_3) + (b_2 + b_3 -q +1)((q^2-q b_3 -1) -  (b_2 + q -b_3 -2)/2) + (m_t- (b_2+b_3q+b_2q^3+b_3q^3))\end{dmath*} to $c$. 
\end{itemize}

Summing all contributions to $c$ we get the formula given in the statement for this part {\it (3)}.
\end{proof}

As above, if the ideal $J$ introduced in Section \ref{se:dos} is generated by $X^{n+1}-X$, with a  very similar proof to that of the previous theorem, the following result holds.

\begin{teo}
Consider the BCH code $E_\Delta$ over the field $\mathbb{F}_{q^2}$ with extension degree 2 and length $q^4$,  where  $\Delta = \Delta (t) = \cup_{j=0}^t  I_{m_j} \subseteq \mathcal{H}$. Set $(b_0,b_1,b_2,b_3)$ the $q$-adic expression of $m_t$. Then the parameters of the corresponding EAQECC are

 $$\left[\left[q^4,q^4-2\sum_{j=0}^{t}  i_{m_j}+c,\ge m_{t+1}+1;c\right]\right]_q,$$ where 
\begin{enumerate}
\item If $m_t<\frac{q^4-1}{q+1}$, then $c=0$.
\item If $\frac{q^4-1}{q+1}\le m_t <q^3+q$, then $c=1+4\left(m_t - \frac{q^4-1}{q+1}\right)$.

\item If $q^3+q\le m_t$, then

\begin{dmath*}
c=  4q(q-1)    + b_3^2  + \max\{0,2(b_2+b_3-q+2)-1\}   \\ \hiderel{+} 4\left( \sum_{j=1}^{b_3-1} \left[ (q-1-j)(q-j)j + (j+1)(q^2 - qj -1) - \frac{(2(q-1)-j)(j+1)}{2} \right]   \\ \hiderel{+}
\delta\left(b_2b_3(q-b_3)+ b_3(b_1-b_3)+\max\{0,b_0+b_3-q+1)\}\right)   \\ \hiderel{+}  (1-\delta)\left[ (q-1-b_3)b_3(q-b_3) + (b_2 + b_3 -q +1)((q^2-q b_3 -1) -   (b_2 + q -b_3 -2)/2) + (m_t- (b_2+b_3q+b_2q^3+b_3q^3)) \right] \right) ,\end{dmath*}where $\delta=1$ if $b_2 + b_3 <q-1$ and $\delta = 0$ otherwise.



 

\end{enumerate}
\end{teo}  

This family of codes contains codes with good parameters. As a sample, to finish this section and the paper, we present several tables, Tables \ref{ta:uno}, \ref{ta:dos}, \ref{ta:tres} and \ref{ta:cuatro}, containing parameters of codes obtained with our formulae, over different finite fields, that exceed the Gilbert-Varshamov (GV) bound for EAQECCs \cite{QINP3}.

\begin{table}
\begin{center}
\begin{tabular}{|c|c|c|c||c|c|c|c|}
  \hline
  $n$ & $k$ & $d\ge$  & $c$ & $n$ & $k$ & $d \ge$ & $c$ \\
  \hline
  80 & 75 & 3 & 1 & 80 & 71 & 4 & 1 \\
  80 & 67 & 5 & 1 & 80 & 63 & 6 & 1 \\
  80 & 47 & 11 & 1 & 80 & 45 & 12 & 1 \\
  80 & 41 & 13 & 1 & 80 & 37 & 14 & 1 \\
  80 & 21 & 18 & 1 & 80 & 17 & 21 & 1 \\
  80 & 16 & 22 & 2 & 80 & 16 & 23 & 6 \\
  80 & 16 & 24 & 10 & 80 & 16 & 25 & 14 \\
  \hline
\end{tabular}
\caption{Parameters of EAQECCs over $\mathbb{F}_3$}\label{ta:uno}
\end{center}
\end{table}

 \begin{table}
\begin{center}
\begin{tabular}{|c|c|c|c||c|c|c|c|}
  \hline
   $n$ & $k$ & $d\ge$  & $c$ & $n$ & $k$ & $d \ge$ & $c$ \\
  \hline
  255 & 250 & 3 & 1 & 255 & 246 & 4 & 1 \\
  255 & 242 & 5 & 1 & 255 & 238 & 6 & 1 \\
  255 & 194 & 18 & 1 & 255 & 192 & 19 & 1 \\
  255 & 176 & 23 & 1 & 255 & 172 & 24 & 1 \\
  255 & 148 & 30 & 1 & 255 & 144 & 31 & 1 \\
  255 & 136 & 35 & 1 & 255 & 134 & 36 & 1 \\
  255 & 118 & 40 & 1 & 255 & 114 & 41 & 1 \\
  255 & 81 & 54 & 6 & 255 & 81 & 55 & 10 \\
  255 & 81 & 56 & 14 & 255 & 81 & 57 & 18 \\
  255 & 81 & 64 & 46 & 255 & 81 & 69 & 50 \\
  255 & 79 & 70 & 50 & 255 & 75 & 71 & 50 \\
  \hline
\end{tabular}
\caption{Parameters of EAQECCs over $\mathbb{F}_4$}\label{ta:dos}
\end{center}
\end{table}

\begin{table}
\begin{center}
\begin{tabular}{|c|c|c|c||c|c|c|c|}
  \hline
 $n$ & $k$ & $d\ge$  & $c$ & $n$ & $k$ & $d \ge$ & $c$ \\
  \hline
   624 & 619 & 3 & 1 & 624 & 615 & 4 & 1 \\
   624 & 611 & 5 & 1 & 624 & 607 & 6 & 1 \\
   624 & 527 & 27 & 1 & 624 & 525 & 28 & 1 \\
   624 & 521 & 29 & 1 & 624 & 517 & 30 & 1 \\
   624 & 513 & 31 & 1 & 624 & 509 & 32 & 1 \\
   624 & 489 & 37 & 1 & 624 & 485 & 38 & 1 \\
   624 & 395 & 63 & 1 & 624 & 391 & 64 & 1 \\
   624 & 273 & 97 & 1 & 624 & 269 & 98 & 1 \\
   624 & 261 & 100 & 1 & 624 & 257 & 105 & 1 \\
   624 & 256 & 107 & 6 & 624 & 256 & 108 & 10 \\
   624 & 256 & 113 & 30 & 624 & 256 & 114 & 34 \\
   624 & 256 & 115 & 38 & 624 & 256 & 116 & 42 \\
   \hline
 \end{tabular}
\caption{Parameters of EAQECCs over $\mathbb{F}_5$}\label{ta:tres}
\end{center}
\end{table}

\begin{table}
\begin{center}
\begin{tabular}{|c|c|c|c||c|c|c|c|}
  \hline
 $n$ & $k$ & $d\ge$  & $c$ & $n$ & $k$ & $d \ge$ & $c$ \\
  \hline
  2400 & 2399 & 2 & 1 &  2400 & 2395 & 3 & 1 \\
  2400 & 2391 & 4 & 1 &  2400 & 2387 & 5 & 1 \\
  2400 & 2383 & 6 & 1 &  2400 & 2379 & 7 & 1 \\
  2400 & 2207 & 51 & 1 &  2400 & 2205 & 52 & 1 \\
  2400 & 2201 & 53 & 1 &  2400 & 2197 & 54 & 1 \\
  2400 & 2141 & 68 & 1 &  2400 & 2137 & 69 & 1 \\
  2400 & 1919 & 126 & 1 &  2400 & 1915 & 127 & 1 \\
  2400 & 1911 & 128 & 1 &  2400 & 1907 & 129 & 1 \\
  2400 & 1907 & 130 & 1 &  2400 & 1899 & 31 & 1 \\
  2400 & 1829 & 152 & 1 &  2400 & 1825 & 153 & 1 \\
  2400 & 1781 & 164 & 1 &  2400 & 1777 & 165 & 1 \\
  2400 & 1737 & 175 & 1 &  2400 & 1733 & 176 & 1 \\
  \hline
\end{tabular}
\caption{Parameters of EAQECCs over $\mathbb{F}_7$}\label{ta:cuatro}
\end{center}
\end{table}

Most of the EAQECCs in the literature (see Section \ref{se:uno}) are binary or $q$-ary with length smaller than or equal to $q^2$. Hence, we cannot compare our codes with them. Other articles about BCH codes consider just concrete subfamilies. To the best of our knowledge, the only article with EAQECCs having the same length as ours is \cite{Quian2}. There the authors provide a few codes with length $q^4-1$. Indeed, for a given finite field $\mathbb{F}_q$, they give  two  codes, with parameters $[[q^4-1, q^4 -1 - 5, 3; 1]]_q$ and $[[q^4-1, q^4 -1 - 7, 4; 1]]_q$.  These codes are constructed from almost MDS constacyclic codes using the Hartmann-Tzeng bound. We notice that those with minimum distance 3 are contained in the set of codes presented in this section but the ones with minimum distance 4 are not contained in the codes presented in this section, and have better parameters.  Notice also that the cyclotomic cosets presented in this section are consecutive to bound the minimum distance with the BCH bound but the codes with miminum distance 4 in \cite{Quian2} are constructed from non-consecutive cyclotomic cosets.


\end{document}